\documentclass[final,journal,twoside]{IEEEtranTCOM}
\usepackage[hypertexnames=false,colorlinks]{hyperref}
\usepackage{cite}

\ifCLASSINFOpdf
\else
   \usepackage[dvips]{graphicx}
\fi

\usepackage[cmex10]{amsmath}
\usepackage{bm}
\usepackage{dsfont}
\usepackage{amssymb}
\usepackage{amsthm}
\usepackage{commath}

\usepackage{algorithmicx}
\usepackage{algpseudocode}
\usepackage{framed}

\usepackage[font=footnotesize]{subfig}
\usepackage{url}
\usepackage{breakurl}
\usepackage{xcolor}

\usepackage[capitalize]{cleveref}
\crefname{equation}{}{}

\newtheorem{lemma}{Lemma}[section]
\newtheorem{proposition}{Proposition}[section]
\theoremstyle{definition}
\newtheorem{definition}{Definition}[section]

\begin{document}
%
\title{Joint Millimeter-Wave Fronthaul and OFDMA Resource Allocation in Ultra-Dense CRAN}
%
%
%

\author{Reuben~George~Stephen~\IEEEmembership{Student~Member,~IEEE,} and Rui~Zhang~\IEEEmembership{Fellow,~IEEE}
\thanks{\copyright 2017 IEEE. Personal use of this material is permitted, but republication/redistribution requires IEEE permission. See http://www.ieee.org/publications\_standards/publications/rights/index.html for more information.}
\thanks{This work was supported in part by the National University of Singapore under Research Grant R-263-000-B46-112.}
\thanks{R. G. Stephen is with the NUS Graduate School for Integrative Sciences and Engineering, and also with the Department of Electrical and Computer Engineering, National University of Singapore (NUS), Singapore 117456 (e-mail: reubenstephen@u.nus.edu).}
\thanks{R. Zhang is with the Department of Electrical and Computer Engineering, NUS, Singapore 117583, and also with the Institute for Infocomm Research, Agency for Science, Technology  and Research, Singapore 138632 (e-mail: elezhang@nus.edu.sg).}
}
%
%

\markboth{IEEE Transactions on Communications}%
{Accepted paper}
%


\maketitle
\begin{abstract}
Ultra-dense~(UD) wireless networks 
and cloud radio access networks~(CRAN) 
are two promising network architectures for the emerging fifth-generation~(5G) wireless communication systems. 
By jointly employing them, a new appealing network solution is proposed in this paper, termed UD-CRAN. In a UD-CRAN, millimeter-wave~(mmWave) wireless fronthaul is preferred for information exchange between the central processor and the distributed remote radio heads~(RRHs), due to its lower cost and higher flexibility in deployment, compared to fixed optical links. 
This motivates our study in this paper on the downlink transmission in a mmWave fronthaul enabled, orthogonal frequency division multiple access~(OFDMA) based UD-CRAN. In particular, the fronthaul is shared among the RRHs via time division multiple access~(TDMA); while the RRHs jointly transmit to the users on orthogonal frequency sub-channels using OFDMA. The joint resource allocation over the TDMA-based mmWave fronthaul and OFDMA-based wireless transmission is investigated to maximize the weighted sum rate of all users. Although the problem is non-convex, we propose a Lagrange duality based solution, which can be efficiently computed with good accuracy. To further reduce the complexity, we also propose a greedy search based heuristic, which achieves close to optimal performance under practical setups. Finally, we show the significant throughput gains of the proposed joint resource allocation approach compared to other benchmark schemes by simulations. 

\end{abstract}

\begin{IEEEkeywords}
Cloud radio access network, orthogonal frequency division multiple access, resource allocation, 
ultra-dense network, millimeter-wave 
fronthaul. 
\end{IEEEkeywords}
%
\section{Introduction}
\IEEEpubidadjcol
\IEEEPARstart{T}{he explosion} of wireless data traffic in recent years has led to a demand for a 1000-fold increase in the capacity of the future fifth-generation (5G) wireless communication networks~\cite{andrews-etal2014what}. To achieve this end, increasing the number of cellular base stations~(BSs) deployed to serve a given area, also known as network densification, 
is foreseen to be 
necessary~\cite{andrews-etal2014what,bhushan-etal2014network}. In such ultra dense~(UD) networks, the number of BSs deployed in a given area can be comparable to, or even exceed the number of users~\cite{andrews-etal2014what,bhushan-etal2014network}.\footnote{A UD network can also be defined solely in terms of the BS density, which can be up to $40-50~\text{BSs}/\text{km}^2$~\cite{ge-etal20165g}.} 
On the other hand, a cloud radio access network~(CRAN), where the conventional BSs are replaced by low-power and low-complexity 
remote radio heads~(RRHs) that are coordinated by a central processor~(CP), provides a new cost-effective way to achieve network densification~\cite{andrews-etal2014what}. 
In CRAN, joint signal processing for a cluster of RRHs and their served users can be performed at the CP, which leads to increased spectral and energy efficiency, via centralized resource allocation~\cite{park-etal2013joint,zhao-etal2013coordinated,zhou-yu2014optimized,zhuang-lau2014backhaul,dai-yu2014sparse,shi-etal2014group,luo-etal2015downlink,liu-zhang2015optimized,liu-zhang2015downlink,tao-etal2016content,jain-etal2016backhaul,shi-etal2016smoothed,fan-etal2016dynamic,liu-etal2015joint}. In order to enable CRAN, a cluster of RRHs need to communicate with their associated CP for information exchanges via high-speed links called the fronthaul. The RRHs can be either simple relay nodes without encoding/decoding capability, or can be similar to the BSs in conventional cellular networks with baseband processing capability~\cite{bi-etal2015wireless}. In this paper, we focus on the downlink transmission in a CRAN, where the CP forwards the user messages to the RRHs via the fronthaul,\footnote{When the RRHs have encoding/decoding capability, the links between the RRHs and CP are also referred to as ``backhaul", as in conventional cellular networks. However, we use ``fronthaul", following the terminology for CRAN.} while the RRHs decode, and then re-encode and cooperatively transmit the information to the users. 
\IEEEpubidadjcol

Combining the idea of network densification with centralized joint processing leads to a powerful new network architecture that can support wireless connectivity of ultra high throughput, termed UD-CRAN~\cite{shi-etal2015large}. Traditionally, the fronthaul links in CRAN are provisioned using optical fibers or high-speed Ethernet, with each RRH having a dedicated link to the CP. However, in UD-CRAN, where the RRHs are large in number and may be at locations that are difficult to reach by laying fibers or wires, providing such dedicated wired links between individual RRHs and the CP is not always feasible. Thus, to achieve practically scalable cost and complexity, a millimeter wave~(mmWave) wireless fronthaul is desirable for UD-CRAN, as it is cost-effective, flexible and easier to deploy compared to wired fronthauls~\cite{bojic-etal2013advanced,dehos-etal2014mmwave}. 
The availability of largely unused bandwidth in the mmWave frequencies, especially in the 70--80~GHz E-band, and commercial equipment based on highly directional antennas for transmission and reception in this band~\cite{bojic-etal2013advanced,dehos-etal2014mmwave,johnson2013mobile,benedetto-etal2013huawei,fujitsu} makes it possible to realize a mmWave fronthaul in practice. However, even at mmWave frequencies, the total bandwidth available for the fronthaul can be much less compared to that of commercial fiber links. For example, current generation fronthaul equipment operating in the mmWave E-band can support rates of a few gigabits per second~(Gbps) over bandwidths of around 250~MHz, while fiber-based fronthaul can typically support rates of several hundreds of Gbps~\cite{bojic-etal2013advanced,dehos-etal2014mmwave,johnson2013mobile}. This means that the capacity constraints imposed by the mmWave fronthaul are much more stringent than in the case of wired fronthauls. On the other hand, the mmWave fronthaul allows a more flexible allocation of resources such as bandwidth and time among the RRHs in their communications with the CP, compared to a wired fronthaul, where the capacity of each link is in general fixed. Since the service provider knows the locations of the RRHs and provides fronthaul infrastructure, and due to the highly directional antennas used in mmWave frequencies, line-of-sight~(LoS) communication is possible between the CP and each RRH. However, since the free-space path loss at mmWave frequencies is much more severe compared to traditional microwave communication at lower frequencies, directional transmission is necessary to compensate for the path loss over large distances, and also provides the additional advantage of reduced interference. Thus, although LoS communication can ensure an almost error-free communication between the CP and each RRH, the capacity of the mmWave LoS channel to each RRH can be different based on its distance from the CP, which calls for its judicious use. Motivated by the practical considerations listed above, we consider in this paper a time division multiple access~(TDMA)-based mmWave fronthaul, where the CP transmits to multiple RRHs over their LoS channels over orthogonal time slots. The average rate on the fronthaul to each RRH is thus determined by the time allotted for transmission by the CP to the RRH and the capacity of the mmWave LoS link. 

Most of the existing work on CRAN and coordinated transmission by multiple BSs considers narrow-band transmission shared by multiple users~\cite{park-etal2013joint,zhao-etal2013coordinated,zhou-yu2014optimized,zhuang-lau2014backhaul,dai-yu2014sparse,shi-etal2014group,luo-etal2015downlink,liu-zhang2015optimized,liu-zhang2015downlink,tao-etal2016content,jain-etal2016backhaul,shi-etal2016smoothed,fan-etal2016dynamic,mosleh-etal2016proportional}. However, the orthogonal frequency division multiple access~(OFDMA) based multiuser transmission is more appealing for the high-throughput demanding wireless networks such as 5G.  
For the case with wired backhauls, BS coordination in an OFDMA-based cellular network with backhaul constraints is studied in~\cite{chowdhery-etal2011cooperative,mehryar-etal2012dynamic}. In these papers, each BS could share user data with neighboring BSs on different sub-sets of sub-channels~(SCs), which are chosen heuristically, while no centralized processing is considered. With centralized processing, a joint power and fronthaul rate allocation problem is studied for the uplink transmission in an OFDMA-based CRAN in~\cite{liu-etal2015joint}, where each RRH performs scalar quantization on each SC and forwards the quantized data to the CP for joint decoding. Unlike all the above mentioned prior work that assume dedicated wired fronthaul/backhaul links that connect the RRHs to each other or to the CP, we consider in this paper a new setup of UD-CRAN with a mmWave fronthaul that is shared among the RRHs for communicating with the CP. We focus on the downlink transmission in UD-CRAN and study the optimization for joint resource allocation over the TDMA-based mmWave fronthaul and the OFDMA-based wireless access to maximize the users' weighted sum-throughput from the CP. 
\begin{figure}[h]
\centering
\includegraphics[width=\linewidth]{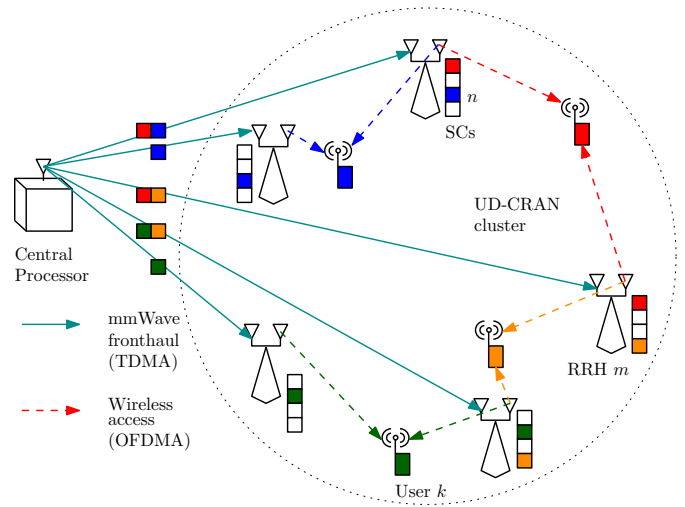}
\caption{Schematic of OFDMA-based UD-CRAN with mmWave fronthaul.}\label{F:SysModel}
\end{figure}

In a broad sense, CRAN can also be viewed as a cooperative relay network~\cite{nosratinia-etal2004cooperative} where multiple relays~(RRHs) cooperatively forward the signal from one source to one or more destinations in the downlink communication. In conventional resource allocation problems considered earlier for such networks~\cite{ng-yu2007joint,kadloor-adve2010relay}, the destination receives multiple copies of the signal both from the source and the relays to maximize the received signal-to-noise ratio~(SNR). In contrast, in this paper, we consider a joint TDMA/OFDMA resource allocation problem over both the mmWave fronthaul~(the first hop) and the wireless access links~(the second hop), which is a new and more general design problem not yet considered in the literature.

In particular, we consider a single cluster of $M$ RRHs in an OFDMA-based UD-CRAN with $N$ orthogonal SCs and $K$ users, where all the RRHs and user terminals are each equipped with a single antenna as shown in~\cref{F:SysModel}. The CP first sends the users' message bits to the RRHs via a mmWave wireless fronthaul shared among them using TDMA. On each SC, due to the limited wireless fronthaul capacity, in general only a sub-set of the RRHs are scheduled to receive the message for one particular user from the CP, which then encode the message using OFDMA and cooperatively transmit it to the assigned user. 
The joint transmission by the selected RRHs on each SC leads to a higher coherent-combining~(transmit beamforming) gain, and hence helps increase the transmission rate to the user assigned to the SC. However, this rate improvement must be supported by the fronthaul of all the RRHs participating in the joint transmission. In addition, the achievable rate on each SC depends on the transmit power levels allocated by each of the transmitting RRHs, as well as their wireless channels to the user assigned to the SC. This thus calls for a new joint resource allocation and transmission scheduling in both the TDMA-based mmWave fronthaul and OFDMA-based wireless transmission. The main results of this paper are summarized as follows.  
\begin{itemize}
\item We study a joint resource allocation problem in UD-CRAN, including the fronthaul TDMA time allocation for different RRHs, the selected SCs and their transmit power allocation for OFDMA transmission at different RRHs, as well as the selected orthogonal SCs for each of the users to maximize their weighted sum rate~(WSR) in the downlink transmission. To the best of our knowledge, this work is the first that considers joint mmWave fronthaul and wireless access transmission optimization in a hybrid TDMA/OFDMA network setup. This problem, however, is combinatorial and non-convex, and incurs an exponential complexity of $O\left(\left(K2^M\right)^N\right)$ if a simple exhaustive search of the optimal solution is conducted. Evidently, this complexity is not practically affordable in a system with large values of $M$ and/or $N$. 
\item We thus propose a Lagrange duality based algorithm, which can achieve the optimal solution asymptotically when the number of SCs is large, with a reduced complexity of $O\left(NK2^M\right)$. 
\item We show that on each SC, the received SNR at the user under the optimal power allocation by the RRHs is a submodular set function of the set of selected RRHs on that SC. Motivated by this result, we propose a greedy algorithm based suboptimal solution, with a reduced complexity of $O\left(NKM^2\right)$, which is shown to be able to achieve close-to-optimal throughput performance under various practical setups by simulations. 
\item Finally, we compare the proposed solutions to other benchmark schemes by simulations, which show that they can achieve significant throughput gains, thanks to the new joint mmWave fronthaul and OFDMA transmission optimization.
\end{itemize} 

The rest of this paper is organized as follows. \cref{Sec:SysMod} describes the system model of UD-CRAN with OFDMA-based wireless access and TDMA-based mmWave fronthaul. The joint resource allocation problem for WSR maximization is presented in~\cref{Sec:ProbForm} and the proposed solutions are given in~\cref{Sec:PropSol}. \cref{Sec:SimResults} presents simulation results comparing the proposed solutions with other benchmark schemes in terms of achievable sum rate. Finally, \cref{Sec:Conc} concludes the paper.

{\it Notation}: In this paper, scalars are denoted by lower-case letters, e.g., $x$, while vectors are denoted by bold-face lower-case letters, e.g., $\bm x$. The set of real numbers, non-negative real numbers and complex numbers are denoted by $\mathds{R}$, $\mathds{R}_+$ and $\mathds{C}$, respectively. Similarly, $\mathds{R}^{x\times 1}$, $\mathds{R}_+^{x\times 1}$ and $\mathds{C}^{x\times 1}$ denote the corresponding spaces of $x$-dimensional column vectors. For $x\in\mathds{R}$, $[x]^+\triangleq\max\{x,0\}$. Also, $\lceil x\rceil$  denotes the smallest integer greater than or equal to $x$, and $\lfloor x\rfloor$ denotes the largest integer less than or equal to $x$. 
For $x\in\mathds{C}$, $\left|x\right|\geq 0$ denotes the magnitude of $x$ and $\angle x\in[0,2\pi)$ denotes its phase angle. For a vector $\bm x$, $\bm x^\mathsf{T}$ denotes its transpose, and $\left\|\bm x\right\|$ denotes its Euclidean norm. Vectors with all elements equal to $1$ and $0$ are denoted by $\bm 1$ and $\bm 0$, respectively. For $\bm x,\bm y\in\mathds{R}^{M\times 1}$, $\bm x\succeq\bm y$ denotes component-wise inequalities, i.e., $x_i\geq y_i\enspace\forall i=1,\dotsc,M$. In addition,  $\mathsf{diag}(\begin{matrix}x_1&\cdots&x_M\end{matrix})$ denotes an $M\times M$ diagonal matrix with diagonal elements given by $x_1,\dotsc,x_M$. For a finite set $\mathcal{S}$, $\left|\mathcal{S}\right|$ denotes its cardinality and $2^\mathcal{A}$ denotes the set of all subsets of $\mathcal{A}$. Finally, $\mathcal{CN}\left(0,\sigma^2\right)$ denotes a circularly symmetric complex Gaussian~(CSCG) random variable with mean $0$ and variance $\sigma^2$, and the symbol $\sim$ is used to mean  ``distributed as". 
\section{System Model}\label{Sec:SysMod}
We study the downlink transmission in a single UD-CRAN cluster consisiting of $M$ single-antenna RRHs, and $K$ single-antenna users, as shown in~\cref{F:SysModel}. Let $\mathcal{M}=\{1,\dotsc,M\}$ denote the set of RRHs, and $\mathcal{K}=\{1,\dotsc,K\}$ denote the set of users. We consider that the RRHs receive the users' data from the CP via mmWave communications by sharing a given spectrum of bandwidth $W$~Hz centered at a frequency of $73$~GHz, using TDMA. The links between the CP and each RRH are LoS, with a free-space path loss given by $69.7 + 24\log_{10}\left(D_m\right)$~dB~\cite{rappaport-etal2015wideband,maccartney-etal2016millimeter}, where $D_m$ in meters~(m) is the distance between the CP and the RRH $m$. We further consider that all the RRHs encode their received data using OFDMA and then cooperatively transmit to the users in the downlink. The wireless access transmission to the users takes place over a multipath channel of bandwidth $B$~MHz centered at a frequency of $2$~GHz, which is equally divided into $N$ orthogonal frequency SCs following the Third Generation Partnership Project~(3GPP) Long Term Evolution-Advanced~(LTE-A) standard~\cite{3gpp36211}. As the mmWave fronthaul and wireless access transmissions are over different frequency bands, the transmission between the CP and the RRHs, and that between the RRHs and the users, can take place simultaneously without interfering with each other. 

Let $\mathcal{N}=\{1,\dotsc,N\}$ denote the set of orthogonal SCs, and let $\nu_{k,n}$ indicate whether user $k$ is assigned to SC $n$, i.e., 
\begin{align}
\nu_{k,n}=\begin{cases}1&\text{if user }k\text{ is assigned to SC }n\nonumber\\
0&\text{otherwise.}\end{cases}
\end{align} 
Also define $\bm\nu_n\triangleq\begin{bmatrix}\nu_{1,n}&\cdots&\nu_{K,n}\end{bmatrix}^\mathsf{T}\in\{0,1\}^{K\times 1}$ as the user association vector at SC $n$. According to OFDMA, each SC $n\in\mathcal{N}$ is assigned to at most one user in downlink transmission, and thus, $\bm 1^\mathsf{T}\bm\nu_n\leq 1,\enspace\forall n\in\mathcal{N}$.
Then, the set of SCs assigned to user $k$, denoted by $\mathcal{N}_k\subseteq\mathcal{N}$, is given by $\mathcal{N}_k=\left\{n\middle|\nu_{k,n}=1\right\}$, where $\mathcal{N}_j\cap\mathcal{N}_k=\emptyset,~\forall j\neq k,~j,k\in\mathcal{K}$. 

Since the mmWave fronthaul capacity for each RRH is practically limited, in general it can only receive the data for a selected sub-set of the users from the CP, and then forward them to the selected users in the OFDMA-based downlink transmission. As a result, each RRH $m$ transmits only on the corresponding sub-set of SCs that are assigned to the users whose data are received from the CP. Let 
\begin{align}
\alpha_{m,n}=\begin{cases}1&\text{if RRH }m\text{ transmits data on SC }n\\
0&\text{otherwise.}\end{cases}
\end{align} 
Define $\bm\alpha_n\triangleq\begin{bmatrix}\alpha_{1,n}&\cdots&\alpha_{M,n}\end{bmatrix}^\mathsf{T}\in\{0,1\}^{M\times 1}$ as the RRH selection vector at each SC $n\in\mathcal{N}$. Then, the sub-set of RRHs that transmit on SC $n$ is given by
\begin{align}
\mathcal{A}_n=\left\{m\in\mathcal{M}\middle|\alpha_{m,n}=1\right\},\quad n\in\mathcal{N}.\label{E:SetSelRRH}
\end{align}  
Thus, the RRHs in $\mathcal{A}_n$ cooperatively send the data to the user $k$ assigned to SC $n$, i.e., $\nu_{k,n}=1$. In the following two sub-sections, we present the models for the wireless access and mmWave fronthaul transmissions in detail, respectively. An illustration for them is also given in~\cref{F:SysModel}.
\subsection{Wireless Access via OFDMA}
Let $h_{k,m,n}=\left|h_{k,m,n}\right|e^{\angle h_{k,m,n}}$ denote the complex wireless access channel coefficient to the user $k\in\mathcal{K}$, from RRH $m\in\mathcal{M}$, on SC $n\in\mathcal{N}$. 
We assume that the magnitudes of all the channel coefficients $\left|h_{k,m,n}\right|$'s are known at the CP, e.g., using appropriate training methods~\cite{zhang-etal2016locally}. To achieve coherent signal combining from all RRHs in $\mathcal{A}_n$ at the receiver of the user $k$ assigned to SC $n$, the signal transmitted by RRH $m\in\mathcal{M}$ on SC $n\in\mathcal{N}_k$, can be written as 
\begin{align}
x_{m,n}=\alpha_{m,n}\sqrt{p_{m,n}}e^{-\angle h_{k,m,n}}s_{k,n},\quad m\in\mathcal{M},\enspace n\in\mathcal{N}_k\label{E:TxSym}
\end{align}
where $s_{k,n}\sim\mathcal{CN}(0,1)$ is the information-bearing signal that is assumed to be Gaussian, for user $k$ on SC $n\in\mathcal{N}_k$, and $p_{m,n}\geq 0$ denotes the power allocated by RRH $m$ on SC $n$. If $\alpha_{m,n}=0$, then RRH $m$ does not transmit signals to any user on SC $n$, and $x_{m,n}=0$. In this case, $p_{m,n}$ should also be equal to zero, without loss of generality. From~\eqref{E:TxSym}, it can be seen that if the CP conveys the optimal power allocation $p_{m,n}$ to RRH $m$ that cooperatively transmits on SC $n\in\mathcal{N}_k$, i.e., with $\alpha_{m,n}=1$, then this RRH additionally needs to know only the phase $\angle h_{k,m,n}$ of its channel coefficient to user $k$ for downlink transmission. Thus, the CP is assumed to have knowledge of the magnitudes of all the channel coefficients, while each RRH needs to know the phases of the channel coefficients to the users assigned to it, for the corresponding SCs. 

Let $\bm h_{k,n}=\begin{bmatrix}h_{k,1,n}&\cdots&h_{k,M,n}\end{bmatrix}^\mathsf{T}\in\mathds{C}^{M\times 1}$ denote the complex channel coefficient vector from the $M$ RRHs to the user $k$ on SC $n$. Similarly, let $\bm x_n=\begin{bmatrix}x_{1,n}&\cdots&x_{M,n}\end{bmatrix}^\mathsf{T}\in\mathds{C}^{M\times 1}$ denote the vector of transmitted signals by the $M$ RRHs to the user $k$ on SC $n\in\mathcal{N}_k$, with each component defined in~\eqref{E:TxSym}, and let $\bm p_n=\begin{bmatrix}p_{1,n}&\cdots&p_{M,n}\end{bmatrix}^\mathsf{T}\in\mathds{R}_+^{M\times 1}$ denote the transmit power allocation vector on SC $n$ for the $M$ RRHs. Then, the received signal at the user $k$ on SC $n\in\mathcal{N}_k$ is given by
\begin{align}
y_n=\bm h_{k,n}^\mathsf{T}\bm x_n+ z,\quad n\in\mathcal{N}_k
\end{align} 
where 
$z\sim\mathcal{CN}(0,\sigma^2)$ is the additive white Gaussian noise~(AWGN), and $\sigma^2$ is the receiver noise power, which is assumed to be equal at all users. The SNR on each SC $n\in\mathcal{N}_k$ is thus given by
\begin{align}
\gamma_{k,n}\left(\bm\alpha_n,\bm p_n\right)
&=\frac{\left|\bm h_{k,n}^\mathsf{T}\bm x_n\right|^2}{\sigma^2}\notag\\
&=\frac{1}{\sigma^2}\left(\sum_{m=1}^M\left|h_{k,m,n}\right|\alpha_{m,n}\sqrt{p_{m,n}}\right)^2.
\label{E:SCSNR}
\end{align}
The maximum achievable rate in bits per second~(bps) on SC $n\in\mathcal{N}_k$ for user $k$ is given by 
\begin{align}
r_{k,n}\left(\bm\alpha_n,\bm p_n\right)&=\frac{B}{N}\log_2\left(1+\gamma_{k,n}\left(\bm\alpha_n,\bm p_n\right)\right),
n\in\mathcal{N}_k.\label{E:SCRate}
\end{align}
Next, we present the following result on the concavity of the function $r_{k,n}\left(\bm\alpha_n,\bm p_n\right)$. 
\begin{lemma}\label{L:rknConc}
With given RRH selection $\bm\alpha_n$, $r_{k,n}\left(\bm\alpha_n,\bm p_n\right)$ defined in~\eqref{E:SCRate} is jointly concave with respect to $\left\{p_{m,n}\right\}$, $\forall m$ with $\alpha_{m,n}=1$. 
\end{lemma}
\begin{proof}
Please refer to appendix~\ref{App:ProofrknConc}.
\end{proof}
\Cref{L:rknConc} indicates that on each SC $n$, for all RRHs  with $\alpha_{m,n}=1$, i.e., transmitting with $p_{m,n}\geq 0$, the achievable rate by their cooperative transmission, $r_{k,n}\left(\bm\alpha_n,\bm p_n\right)$, is a jointly concave function over their $p_{m,n}$'s. This is a useful property we will utilize later in solving our proposed resource allocation problem.   
\subsection{mmWave Fronthaul}
The CP transmits the users' data to each of the RRHs over a LoS mmWave wireless fronthaul channel via TDMA, for a fraction of time $0\leq t_m\leq 1$ for RRH $m$ subject to 
\begin{align}
\sum_{m=1}^Mt_m\leq 1.\label{E:TStot}
\end{align}
Let $R_m>0$ denote the maximum fronthaul rate in bps, achievable on the mmWave fronthaul link from the CP to RRH $m$. Then, the average rate at which RRH $m$ can receive data from the CP over the mmWave fronthaul channel is $t_mR_m$. On the other hand, the total rate at which RRH $m$ transmits to the users over all $N$ SCs is given by 
\begin{align}
\sum_{n=1}^{N}\alpha_{m,n}\sum_{k=1}^K\nu_{k,n}r_{k,n}\left(\bm\alpha_n,\bm p_n\right)\label{E:FHReq}
\end{align}
where $r_{k,n}\left(\bm\alpha_n,\bm p_n\right)$ is defined in~\eqref{E:SCRate}. At each RRH $m$, the average rate over the mmWave fronthaul needs to be no smaller than that over the wireless access, i.e.,
\begin{align}
t_mR_m\geq \sum_{n=1}^{N}\alpha_{m,n}\sum_{k=1}^K\nu_{k,n}r_{k,n}\left(\bm\alpha_n,\bm p_n\right)\quad\forall m\in\mathcal{M}.\label{E:TSIneq}
\end{align}
It can be easily shown that the constraints in~\eqref{E:TStot} and~\eqref{E:TSIneq} are both satisfied by all the RRHs $m\in\mathcal{M}$ if and only if 
\begin{align}
\sum_{m=1}^M\frac{1}{R_m}\sum_{n=1}^N\alpha_{m,n}\sum_{k=1}^K\nu_{k,n}r_{k,n}\left(\bm\alpha_n,\bm p_n\right)\leq 1.\label{E:TSComb}
\end{align}
Notice that with a given feasible power allocation $\left\{\bar{\bm p}\right\}_{n\in\mathcal{N}}$, user-SC association $\left\{\bar{\bm\nu}_n\right\}_{n\in\mathcal{N}}$, and RRH selection $\left\{\bar{\bm\alpha}_n\right\}_{n\in\mathcal{N}}$ which jointly satisfy the constraint in~\eqref{E:TSComb}, the corresponding TDMA time allocation $\bar{t}_m$ on the mmWave fronthaul can be obtained for each RRH $m$ using the relation, 
\begin{align}
\bar{t}_m=\frac{1}{R_m}\sum_{n=1}^{N}\bar{\alpha}_{m,n}\sum_{k=1}^K\bar{\nu}_{k,n}r_{k,n}\left(\bar{\bm\alpha}_n,\bar{\bm p}_n\right),\quad m\in\mathcal{M}.\label{E:MinTime}
\end{align} 
In the next section, we formulate the joint mmWave fronthaul and OFDMA downlink resource allocation optimization problem for the UD-CRAN. 
\section{Problem Formulation}\label{Sec:ProbForm}
We aim to maximize the WSR of all users, by jointly optimizing the user-SC associations $\left\{
\bm\nu_n\right\}_{n\in\mathcal{N}}$, RRH-SC selections $\left\{
\bm\alpha_n\right\}_{n\in\mathcal{N}}$, and  power allocations of all RRHs at all SCs $\left\{
\bm p_n\right\}_{n\in\mathcal{N}}$, subject to the wireless fronthaul constraint~\eqref{E:TSComb}, and the total power constraint at each individual RRH denoted by $\bar{P}_m,~m\in\mathcal{M}$. Let $\omega_k\geq 0,~k\in\mathcal{K}$, denote the rate weight assigned to user $k$. Then, the problem is formulated as 
\begin{subequations}
\label{P:Main}
\begin{align}
\mathop{\mathrm{maximize}}_{\left\{\bm p_n,\bm\alpha_n,\bm\nu_n\right\}_{n\in\mathcal{N}}}&\enspace\sum_{k=1}^K\omega_k\sum_{n=1}^N\nu_{k,n}r_{k,n}\left(\bm\alpha_n,\bm p_n\right)\tag{\ref*{P:Main}}\\
\mathrm{subject}~\mathrm{to}\nonumber\\
&\enspace\eqref{E:TSComb}\nonumber\\
&\enspace\sum_{n=1}^N p_{m,n}\leq \bar{P}_m\quad\forall m\in\mathcal{M}\label{C:AvgPCMain}\\
&\enspace p_{m,n}\geq 0\quad\forall m\in\mathcal{M},\enspace\forall n\in\mathcal{N}\label{C:PowerMain}\\
&\enspace\alpha_{m,n}\in\{0,1\}\quad\forall m\in\mathcal{M},\enspace\forall n\in\mathcal{N}\label{C:RRHSelMain}\\
&\enspace\bm 1^\mathsf{T}\bm\nu_n\leq 1\quad\forall n\in\mathcal{N}\label{C:USCSUMain}\\
&\enspace\nu_{k,n}\in\{0,1\}\quad\forall k\in\mathcal{K},\enspace\forall n\in\mathcal{N}.\label{C:SCAMain}
\end{align}
\end{subequations}
Problem~\eqref{P:Main} is non-convex due to the integer constraints~\cref{C:RRHSelMain,C:SCAMain}, as well as the coupled variables in both the objective function as well as the constraint~\eqref{E:TSComb}. An exhaustive search over all possible user-SC associations and RRH-SC selections requires $O\left(\left(2^MK\right)^N\right)$ operations, which can be prohibitive for large values of $M$ or $N$. Notice that even if the user association $\bm\nu_n$ and the RRH selection $\bm\alpha_n$ on each SC $n$ is given, the left-hand-side of constraint~\eqref{E:TSComb} is a non-negative weighted sum of concave functions $r_{k,n}\left(\bm\alpha_n,\bm p_n\right)$ according to~\Cref{L:rknConc}, and hence concave, which makes the constraint~\eqref{E:TSComb} still non-convex. Thus, problem~\eqref{P:Main} is non-convex even if the user-SC associations and the RRH-SC selections are all fixed. 
\section{Proposed Solutions}\label{Sec:PropSol}
\subsection{Optimal Solution}\label{SS:OptSol}
Although problem~\eqref{P:Main} is non-convex, it can be verified that strong duality holds when the number of SCs $N$ goes to infinity, as it satisfies the ``time-sharing" conditions as given in~\cite{yu-lui2006dual}. As $N$ is typically large in practice, we propose to apply the Lagrange duality method to solve problem~\eqref{P:Main} by assuming zero duality gap.\footnote{We emphasize that the optimality of the proposed Lagrange duality based solution is in the asymptotic sense, for sufficiently large $N$. However, as shown in~\Cref{Sec:SimResults}, the duality gap is negligible for $N=128$, and hence the proposed solution is close to optimal for practical values of $N$.} 
Let $\lambda\geq 0$ denote the dual variable associated with constraint~\eqref{E:TSComb}, and $\mu_m\geq 0,~m\in\mathcal{M}$, denote the dual variables for the constraints in~\eqref{C:AvgPCMain}. Also define $\bm \mu\triangleq\begin{bmatrix}\mu_1&\cdots&\mu_M\end{bmatrix}^\mathsf{T}\in\mathds{R}_+^{M\times 1}$. 
Then, the~(partial) Lagrangian of problem~\eqref{P:Main} is given by 
\begin{align}
&L\left(\left\{\bm\nu_n,\bm\alpha_n,\bm p_n\right\}_{n\in\mathcal{N}},\lambda,\bm \mu\right)\notag\\
&=\sum_{n=1}^N\sum_{k=1}^K\omega_k\nu_{k,n}r_{k,n}\left(\bm\alpha_n,\bm p_n\right)\notag\\
&\quad-\lambda\left(\sum_{n=1}^N\sum_{m=1}^M\frac{\alpha_{m,n}}{R_m}\sum_{k=1}^K\nu_{k,n}r_{k,n}\left(\bm\alpha_n,\bm p_n\right)-1\right)\notag\\
&\quad-\sum_{m=1}^M\mu_m\left(\sum_{n=1}^N p_{m,n}-\bar{P}_m\right)\notag\\
&=\sum_{n=1}^N L_n\left(\bm\nu_n,\bm\alpha_n,\bm p_n,\lambda,\bm \mu\right)+\lambda+\sum_{m=1}^M\mu_m\bar{P}_m,
\label{E:LagW}
\end{align} 
where 
\begin{align}
L_n\left(\bm\nu_n,\bm\alpha_n,\bm p_n,\lambda,\bm \mu\right)&\triangleq\sum_{k=1}^K\omega_k\nu_{k,n}r_{k,n}\left(\bm\alpha_n,\bm p_n\right)\notag\\
&\quad-\lambda\sum_{m=1}^M\frac{\alpha_{m,n}}{R_m}\sum_{k=1}^K\nu_{k,n}r_{k,n}\left(\bm\alpha_n,\bm p_n\right)\notag\\
&\quad-\sum_{m=1}^M\mu_m p_{m,n}.\label{E:Lagn}
\end{align}
The Lagrange dual function is thus given by  
\begin{subequations}
\label{E:DualFunc} 
\begin{align}
g(\lambda,\bm\mu)=~
\max_{\left\{\bm p_n,\bm\alpha_n,\bm\nu_n\right\}_{n\in\mathcal{N}}}&\enspace
L\left(\left\{\bm\nu_n,\bm\alpha_n,\bm p_n,\right\}_{n\in\mathcal{N}},\lambda,\bm\mu\right)\tag{\ref*{E:DualFunc}}\\
\mathrm{s.t.}
&\enspace\text{\cref{C:USCSUMain,C:PowerMain,C:SCAMain,C:RRHSelMain}}.\notag
\end{align}
\end{subequations}
The maximization problem in~\eqref{E:DualFunc} can be decomposed into $N$ parallel sub-problems, where each sub-problem corresponds to a single SC $n\in\mathcal{N}$, and all of them have the same structure given by
\begin{subequations} 
\label{P:DualFuncn}
\begin{align}
\max_{\bm p_n,\bm\alpha_n,\bm\nu_n}&\enspace
L_n\left(\bm\nu_n,\bm\alpha_n,\bm p_n,\lambda,\bm \mu\right)\tag{\ref*{P:DualFuncn}}\nonumber\\
\mathrm{s.t.}
&\enspace\bm p_n\succeq\bm 0\label{C:VectPNZDFn}\\
&\enspace\bm\alpha_n\in\{0,1\}^{M\times 1}\label{C:BinRRHSelVectDFn}\\
&\enspace\bm 1^\mathsf{T}\bm\nu_n\leq 1\\
&\enspace\bm\nu_n\in\{0,1\}^{K\times 1}
\end{align}
\end{subequations}
where $
L_n\left(\bm\nu_n,\bm\alpha_n,\bm p_n,\lambda,\bm \mu\right)$ is defined in~\eqref{E:Lagn}. 

Now let the user association on SC $n$ be fixed as $\bm\nu_n=\hat{\bm\nu}_n$. If $\hat{\bm\nu}_n=\bm 0$, no user is assigned to SC $n$. In this case, since $\bm\mu\succeq\bm 0$ and $\bm p_n\succeq\bm 0$, the objective of problem~\cref{P:DualFuncn}, as given in~\cref{E:Lagn} is maximized by setting $\bm p_n=\bm 0$, irrespective of the RRH selection $\bm\alpha_n$. Thus, if no user is assigned to SC $n$, the power allocation over all RRHs is zero, as expected. Otherwise, let $\hat{k}_n\in\mathcal{K}$ be the user assigned to SC $n$ so that $\hat{\nu}_{\hat{k}_n,n}=1$ and $\hat{\nu}_{k,n}=0\enspace\forall k\neq\hat{k}_n$. Using this in~\eqref{E:Lagn}, sub-problem~\eqref{P:DualFuncn} on each SC $n$ can be reduced to the following problem
\begin{subequations}
\label{P:DualFuncnFixUA}
\begin{align}
\max_{\bm p_n,\bm\alpha_n}&\enspace\left(\omega_{\hat{k}_n}-\lambda\sum_{m=1}^M\frac{\alpha_{m,n}}{R_m}\right)r_{\hat{k}_n,n}\left(\bm\alpha_n,\bm p_n\right)-\sum_{m=1}^M\mu_m p_{m,n}\tag{\ref*{P:DualFuncnFixUA}}\\
\mathrm{s.t.}&\enspace\text{\cref{C:VectPNZDFn,C:BinRRHSelVectDFn}},\notag
\end{align}
\end{subequations}
which is non-convex due to the integer constraints~\eqref{C:BinRRHSelVectDFn} on $\bm\alpha_n$ and the coupled variables in the objective. However, for a given RRH selection $\tilde{\bm\alpha}_n$, problem~\eqref{P:DualFuncnFixUA} can be written as 
\begin{subequations}
\label{P:DualFuncnFixUARRHSel}
\begin{align}
\max_{\bm p_n\succeq\bm 0}&\enspace\left(\omega_{\hat{k}_n}-\lambda\sum_{m=1}^M\frac{\tilde{\alpha}_{m,n}}{R_m}\right)r_{\hat{k}_n,n}\left(\tilde{\bm\alpha}_n,\bm p_n\right)-\sum_{m=1}^M\mu_m p_{m,n},\tag{\ref*{P:DualFuncnFixUARRHSel}}
\end{align}
\end{subequations}
where $r_{\hat{k}_n,n}\left(\tilde{\bm\alpha}_n,\bm p_n\right)$ is the rate on SC $n$ under RRH selection $\tilde{\bm\alpha}_n$, with user $\hat{k}_n$ assigned to SC $n$, as defined in~\eqref{E:SCRate}.
Then, the optimal power allocation $\tilde{\bm p}_n$ that solves problem~\eqref{P:DualFuncnFixUARRHSel} is given by the following proposition. 
\begin{proposition}\label{Prop:OptPAFixURRHSel}
Let $\bm\alpha_n=\tilde{\bm\alpha}_n$ be fixed. Then, for given $\lambda$ and $\bm\mu$, the optimal power allocation $\tilde{\bm p}_n$ on SC $n$ for problem~\eqref{P:DualFuncnFixUARRHSel} is given by
\begin{align}
\tilde{p}_{m,n}&=\frac{\tilde{\alpha}_{m,n}\left|h_{\hat{k}_n,m,n}\right|^2}{\sigma^2\mu_m^2\left[G_{\hat{k}_n,n}\left(\tilde{\bm\alpha}_n\right)\right]^2}\notag\\
&\quad\cdot\left[\frac{B}{N\ln 2}F_{\hat{k}_n,n}\left(\tilde{\bm\alpha}_n\right)G_{\hat{k}_n,n}\left(\tilde{\bm\alpha}_n\right)-1\right]^+
\label{E:OptPA}
\end{align}
$\forall m\in\mathcal{M}$, where
$F_{\hat{k}_n,n}\left(\bm\alpha_n\right)$, $G_{\hat{k}_n,n}\left(\bm\alpha_n\right)$ are defined as
\begin{align}
F_{\hat{k}_n,n}\left(\bm\alpha_n\right)
&\triangleq~\omega_{\hat{k}_n}-\lambda\sum_{m=1}^M\frac{\alpha_{m,n}}{R_m}
\label{E:FHParam}\\
G_{\hat{k}_n,n}\left(\bm\alpha_n\right)
&\triangleq~\sum_{m=1}^M\frac{\alpha_{m,n}\left|h_{\hat{k}_n,m,n}\right|^2}{\sigma^2\mu_m}.
\label{E:AccParam}
\end{align}
\end{proposition}
\begin{proof}
Please refer to appendix~\ref{A:ProofPAlloc}.
\end{proof}
\Cref{Prop:OptPAFixURRHSel} shows that for given user association $\hat{k}_n$ and RRH selection $\tilde{\bm\alpha}_n$ on each SC $n$, the optimal power allocation has a threshold structure, which allocates zero power to all RRHs on SC $n$ if $F_{\hat{k}_n,n}\left(\tilde{\bm\alpha}_n\right)G_{\hat{k}_n,n}\left(\tilde{\bm\alpha}_n\right)\leq(N\ln 2)/B$. Notice  that $F_{\hat{k}_n,n}\left(\bm\alpha_n\right)$ depends on the fronthaul rates $R_m$'s, while 
$G_{\hat{k}_n,n}\left(\bm\alpha_n\right)$ depends on the wireless access channel gains $\left|h_{\hat{k}_n,n}\right|$'s. Otherwise, if $F_{\hat{k}_n,n}\left(\tilde{\bm\alpha}_n\right)G_{\hat{k}_n,n}\left(\tilde{\bm\alpha}_n\right)>(N\ln 2)/B$, the power allocation on each RRH $m\in\mathcal{M}$, with $\tilde{\alpha}_{m,n}=1$, is proportional to the ratio $\tilde{\alpha}_{m,n}\tfrac{\left|h_{\hat{k}_n,m,n}\right|^2}{\sigma^2\mu_m^2}$, which depends on the wireless access channel gain $\left|h_{\hat{k}_n,m,n}\right|$ on SC $n$ and the dual variable $\mu_m$ corresponding to the transmit power constraint~\eqref{C:AvgPCMain} for RRH $m$. If $\tilde{\bm\alpha}_n=\bm 0$, i.e., no RRH is selected, then $\tilde{\bm p}_n=\bm 0$. When there are no fronthaul constraints, i.e., $\lambda=0$, the optimal power allocation in~\eqref{E:OptPA} reduces to 
\begin{align}
\tilde{p}_{m,n}=\frac{\tilde{\alpha}_{m,n}\left|h_{\hat{k}_n,m,n}\right|^2}{\sigma^2\mu_m^2\left[G_{\hat{k}_n,n}\left(\tilde{\bm\alpha}_n\right)\right]^2}\left[\frac{B\omega_{\hat{k}_n}}{N\ln 2}G_{\hat{k}_n,n}\left(\tilde{\bm\alpha}_n\right)-1\right]^+.
\end{align} 
For the special case when there is only one RRH in the cluster, i.e., $M=1$, the power allocation in~\eqref{E:OptPA} becomes\footnote{Drop the subscript $m$ since $M=1$.}
\begin{align}
\tilde{p}_n=\left[\frac{B}{\mu N\ln 2}\left(\omega_{\hat{k}_n}-\frac{\lambda}{R}\right)-\frac{\sigma^2}{\left|h_{\hat{k}_n,n}\right|^2}\right]^+,\label{E:OptPASingleRRH}
\end{align} 
which has the same form as the well-known water-filling solution~\cite{goldsmith2005wireless}, but in general with different water levels on different SCs $n\in\mathcal{N}$, which are determined by the user $\hat{k}_n$ assigned to SC $n$, through $\omega_{\hat{k}_n}$. Also notice that the water-level on SC $n$ can be negative if $\omega_{\hat{k}_n}\leq\lambda/R$; in this case, no power should be allocated to SC $n$. 

From~\cref{Prop:OptPAFixURRHSel}, it can be seen that the optimal power allocation is zero whenever a particular RRH $m$ is not selected on SC $n$, i.e., if $\tilde{\alpha}_{m,n}=0$, it implies that $\tilde{p}_{m,n}=0$, as expected. However, the opposite is not necessarily true, i.e., if it turns out that $\tilde{p}_{m,n}=0$ for some given RRH selection $\tilde{\bm\alpha}_n$, it cannot be inferred that $\tilde{\alpha}_{m,n}=0$ in the optimal RRH selection. 
This is due to the fact that the power allocation $\tilde{\bm p}_n$ given by~\eqref{E:OptPA} depends on the entire RRH selection $\tilde{\bm\alpha}_n$, and even if $\tilde{p}_{m,n}=0$, it cannot be known beforehand whether this particular RRH selection and power allocation is the one that maximizes the objective in~\eqref{P:DualFuncnFixUA}. Thus,~\cref{Prop:OptPAFixURRHSel} only gives the optimal power allocation when the RRH selection is known, and when $M>1$, an exhaustive search over all possible vectors $\bm\alpha_n\in\{0,1\}^M$ is required to find the RRH selection and corresponding power allocation that maximizes the objective in~\eqref{P:DualFuncnFixUA}. When $M=1$, however, the RRH selection on each SC is implicitly decided by whether the optimal power allocation given by~\eqref{E:OptPASingleRRH} is zero or not. 

For given dual variables $\lambda$ and $\bm\mu$, problem~\eqref{P:DualFuncn} can be solved optimally using~\cref{Prop:OptPAFixURRHSel} as follows. First, fix the user on SC $n$ as $\hat{k}_n\in\mathcal{K}$. Then, for each of the $2^M$ possible RRH selections, compute the optimal power allocation $\tilde{\bm p}_n$ using~\eqref{E:OptPA}, and choose the optimal RRH selection $\hat{\bm\alpha}_n$ for the user $\hat{k}_n$ as the one that maximizes the objective of problem~\eqref{P:DualFuncnFixUA} with the corresponding power allocation $\hat{\bm p}_n$ given by~\eqref{E:OptPA}. Then the optimal user association $\bar{\bm\nu}_n$ on SC $n$ can be found by choosing the user $\bar{k}_n$ that maximizes the objective of problem~\eqref{P:DualFuncn}, with its corresponding optimal RRH selection and power allocation computed before, and denoted by $\bar{\bm\alpha}_n$ and $\bar{\bm p}_n$. 

Now, the dual problem for~\eqref{P:Main} is given by 
\begin{align}
\min_{\lambda\geq 0,\bm\mu\succeq\bm 0}g(\lambda,\bm\mu),\label{E:DualProb}
\end{align}
which is convex and can be solved efficiently, e.g., using the ellipsoid method~\cite{boyd2014ellipsoid} to find the optimal dual variables $\lambda^\star$ and $\bm\mu^\star$. Then, the optimal solution to problem~\eqref{P:DualFuncn} is given by $\left(\bar{\bm\nu}_n,\bar{\bm\alpha}_n,\bar{\bm p}_n\right)$, computed as outlined above at the optimal dual variables $\lambda^\star$ and $\bm\mu^\star$. The algorithm for solving problem~\eqref{P:Main} can thus be summarized as given in~\cref{A:Overall}. 
\begin{table}[h]
\caption{Algorithm for problem~\eqref{P:Main}}\label{A:Overall}
\begin{framed}
\begin{algorithmic}[1]
\State Initialization: $\lambda\geq 0$, $\bm\mu\succeq\bm 0$
\Repeat 
\For {each $n\in\mathcal{N}$}
\For {each user association $\hat{\bm\nu}_n$}
\State Find optimal RRH selection $\hat{\bm\alpha}_n$ that maximizes~\eqref{P:DualFuncnFixUA} by searching over $2^M$ possible RRH selections, with optimal power allocation $\hat{\bm p}_n$ given by~\eqref{E:OptPA}\label{AL:ORRHSelP}
\EndFor
\State Find optimal user association $\bar{\bm\nu}_n$ that maximizes~\eqref{P:DualFuncn} with optimal RRH selection $\bar{\bm\alpha}_n$ and power allocation $\bar{\bm p}_n$ computed in~\cref{AL:ORRHSelP}
\EndFor
\State Update dual variables $\lambda$ and $\bm\mu$ using the ellipsoid method
\Until ellipsoid algorithm converges to desired accuracy
\end{algorithmic}
\end{framed}
\end{table}

Finding the optimal RRH selection and power allocation for a given user association involves a search over $2^M$ values, and incurs a complexity of $O\left(2^M\right)$. Subsequently, finding the optimal user association involves a search over $K$ users, and has a complexity of $O(K)$. Each of the $N$ problems~\eqref{P:DualFuncn} can thus be solved incurring a complexity of $O\left(K2^M\right)$. Hence, the dual function $g(\lambda,\bm\mu)$ in~\eqref{E:DualFunc} can be computed for each given pair of $\lambda$ and $\bm\mu$ with an overall complexity of $O\left(NK2^M\right)$. The complexity of the ellipsoid method to find the optimal dual variables depends only on the size of the initial ellipsoid and the maximum length of the sub-gradients over the intial ellipsoid~\cite{boyd2014ellipsoid}. Thus, the overall complexity of solving problem~\eqref{P:Main} is effectively given by $O\left(NK2^M\right)$. For small cluster sizes with $M\leq 5$, this complexity is not very high, and problem~\eqref{P:Main} can be efficiently solved using the algorithm in~\cref{A:Overall}. In the next section we introduce a lower complexity greedy algorithm that still performs very well in practical setups. 
\subsection{Suboptimal Solution}\label{SS:SubOptSubP}
Optimally solving the sub-problem~\eqref{P:DualFuncnFixUA} on each SC for a fixed user association requires an exhaustive search over all possible RRH selections, which is $O\left(2^M\right)$. To avoid this, we propose another way to solve problem~\eqref{P:DualFuncnFixUA} suboptimally in this sub-section, by using a greedy algorithm instead. First, we observe that since~\eqref{E:OptPA} gives the optimal power allocation for a given set $\mathcal{A}_n$ of RRHs that are selected to transmit on SC $n$,  substituting~\eqref{E:OptPA} in~\eqref{E:SCSNR}, the optimal SNR on SC $n$ for a non-zero power allocation can be expressed in terms of $\mathcal{A}_n$, as 
\begin{align}
\tilde{\gamma}_{\hat{k}_n,n}\left(\mathcal{A}_n\right)=\frac{B}{N\ln 2}F_{\hat{k}_n,n}\left(\mathcal{A}_n\right)G_{\hat{k}_n,n}\left(\mathcal{A}_n\right)-1.\label{E:OptSNRSetFunc}
\end{align}
Thus, $\tilde{\gamma}_{\hat{k}_n,n}\left(\mathcal{A}_n\right):2^\mathcal{M}\mapsto\mathds{R}$ in~\eqref{E:OptSNRSetFunc} is a set function that maps the set of all subsets of the set of RRHs $\mathcal{M}$ to a real number. The set functions $F_{\hat{k}_n,n}\left(\mathcal{A}_n\right):2^\mathcal{M}\mapsto\mathds{R}$ and $G_{\hat{k}_n,n}\left(\mathcal{A}_n\right):2^\mathcal{M}\mapsto\mathds{R}_+$ are similarly expressed in terms of $\mathcal{A}_n$ using~\cref{E:FHParam,E:AccParam}. Next, we introduce the following definitions pertaining to set functions. 
\begin{definition}
\label{Def:MonSubMod}
Let $\mathcal{V}$ denote a finite set and $f:2^\mathcal{V}\mapsto\mathds{R}$ be a real-valued set function. Then 
\begin{enumerate}
\item $f$ is \emph{monotone} if and only if 
\begin{align}
f\left(\mathcal{S}\right)\leq f\left(\mathcal{T}\right)\enspace\forall\mathcal{S}\subseteq\mathcal{T}\subseteq\mathcal{V}.\label{E:Mon}
\end{align}
\item $f$ is \emph{submodular} if and only if~\cite{bach2013learning}
\begin{align}
f\left(\mathcal{S}\cup\{i\}\right)-f\left(\mathcal{S}\right)&\geq f\left(\mathcal{S}\cup\{i,j\}\right)-f\left(\mathcal{S}\cup\{j\}\right),\notag\\
&\forall \mathcal{S}\subseteq\mathcal{V},~i\neq j,i,j\in\mathcal{V}\setminus\mathcal{S}.\label{E:SubMod}
\end{align}
In addition, $f$ is \emph{supermodular} if $-f$ is submodular, and $f$ is \emph{modular} if it is both submodular and supermodular.
\end{enumerate}
\end{definition}
Condition~\eqref{E:Mon} implies that a monotone set function is non-decreasing when elements are added to a set. For a submodular function, condition~\eqref{E:SubMod} implies that the incremental difference in its value when a new element is added to a smaller set, is in general, higher than the corresponding difference when the same element is added to a larger set. According to~\cref{Def:MonSubMod}, it is readily observed that the functions $-F_{\hat{k}_n,n}\left(\mathcal{A}_n\right)$ and $G_{\hat{k}_n,n}\left(\mathcal{A}_n\right)$ in~\cref{E:FHParam,E:AccParam} are monotone, while $\tilde{\gamma}_{\hat{k}_n,n}\left(\mathcal{A}_n\right)$ is not. Similarly, it is easy to see from~\eqref{E:SubMod} that both $F_{\hat{k}_n,n}\left(\mathcal{A}_n\right)$ and $G_{\hat{k}_n,n}\left(\mathcal{A}_n\right)$ are modular, while the following proposition shows that $\tilde{\gamma}_{\hat{k}_n,n}\left(\mathcal{A}_n\right)$ is submodular. 
\begin{proposition}\label{Prop:OptSNRSubMod}
The function $\tilde{\gamma}_{\hat{k}_n,n}\left(\mathcal{A}_n\right)$ in~\eqref{E:OptSNRSetFunc} is submodular. 
\end{proposition}
\begin{proof}
Please refer to appendix~\ref{App:ProofOptSNRSubMod}.
\end{proof}
Using~\eqref{E:OptSNRSetFunc}, the objective function of problem~\eqref{P:DualFuncnFixUA} under the optimal power allocation can be expressed as a set function $\tilde{f}\left(\mathcal{A}_n\right):2^\mathcal{M}\mapsto\mathds{R}$ given by
\begin{align}
\tilde{f}\left(\mathcal{A}_n\right)&\triangleq \frac{B}{N}F_{\hat{k}_n}\left(\mathcal{A}_n\right)\log_2\left(1+\left[\tilde{\gamma}_{\hat{k}_n,n}\left(\mathcal{A}_n\right)\right]^+\right)\notag\\
&\quad-\frac{\left[\tilde{\gamma}_{\hat{k}_n}\left(\mathcal{A}_n\right)\right]^+}{G\left(\mathcal{A}_n\right)},\label{E:ObjFUAOptPASet}
\end{align}
and problem~\eqref{P:DualFuncnFixUA} can be expressed as the following combinatorial problem
\begin{align}
\max_{\mathcal{A}_n\subseteq\mathcal{M}}~\tilde{f}\left(\mathcal{A}_n\right).\label{P:DFnFUAOptPRRHSelComb}
\end{align}
For maximizing a monotone, non-negative submodular function under a cardinality constraint, it is well known that a greedy procedure with linear complexity can achieve an objective value at least $(1-1/e)$ of the optimal~\cite{nemhauser1978analy}. For unconstrained maximization of a non-monotone, non-negative submodular function, deterministic and randomized algorithms providing reasonable  worst-case guarantees $1/2$ respectively, are known~\cite{buchbinder2012tight}. However, unlike the optimal SNR $\tilde{\gamma}_{\hat{k}_n,n}\left(\mathcal{A}_n\right)$ in~\eqref{E:OptSNRSetFunc}, the objective function $\tilde{f}\left(\mathcal{A}_n\right)$ of problem~\eqref{P:DFnFUAOptPRRHSelComb} cannot be shown to be submodular; while neither $\tilde{\gamma}_{\hat{k}_n,n}\left(\mathcal{A}_n\right)$ in~\eqref{E:OptSNRSetFunc} or $\tilde{f}\left(\mathcal{A}_n\right)$ in~\eqref{E:ObjFUAOptPASet} are monotone or non-negative. 

Nevertheless, motivated by the submodularity of $\tilde{\gamma}_{\hat{k}_n,n}\left(\mathcal{A}_n\right)$ as shown by~\cref{Prop:OptSNRSubMod}, we propose a suboptimal greedy algorithm with quadratic complexity in the number of RRHs $M$ to solve problem~\eqref{P:DFnFUAOptPRRHSelComb} approximately. We construct a suboptimal set $\check{\mathcal{A}}_n$ for problem~\eqref{P:DFnFUAOptPRRHSelComb} using a greedy algorithm as follows. At the start of an iteration $i$, let $\mathcal{A}_{n,i}$ denote the set of selected RRHs on SC $n$ and $\tilde{f}_i$ denote the objective value of problem~\eqref{P:DFnFUAOptPRRHSelComb} given by~\eqref{E:ObjFUAOptPASet}. Initially, we assume that $\mathcal{A}_{n,1}=\emptyset$, and $\tilde{f}_1=0$. Then, at each iteration $i=1,\dotsc,M$, we first find an RRH $j_i\in\mathcal{M}\setminus\mathcal{A}_{n,i}$, which is not currently selected, and when added to the current set of RRHs $\mathcal{A}_{n,i}$, gives the maximum value of the objective in~\eqref{E:ObjFUAOptPASet} among all the currently un-selected RRHs, i.e., 
\begin{align}
j_i=\mathop{\arg\max}_{\ell\in\mathcal{M}\setminus\mathcal{A}_{n,i}}~\tilde{f}\left(\mathcal{A}_{n,i}\cup\{\ell\}\right).\label{E:MaxRRHji}
\end{align}
Then, if RRH $j_i$ improves the current maximum objective value $\tilde{f}_i$, we add it to the current set of selected RRHs $\mathcal{A}_{n,i}$. That is, if 
\begin{align}
\tilde{f}\left(\mathcal{A}_{n,i}\cup\{j_i\}\right)>\tilde{f}_i
\label{E:NewObjGreater}
\end{align}
holds, the set of selected RRHs is updated as $\mathcal{A}_{n,i+1}=\mathcal{A}_{n,i}\cup\{j_i\}$, and the current maximum objective value is updated as 
\begin{align}
\tilde{f}_{i+1}=\tilde{f}\left(\mathcal{A}_{n,i}\cup\{j_i\}\right).\label{E:UpdCObj}
\end{align}
This procedure is continued until no RRH $j_i$ can be found which satisfies~\eqref{E:NewObjGreater}, or there is no more remaining RRH to be searched, i.e., $i=M$. If the algorithm stops at iteration $i$, the final suboptimal RRH selection is given by $\check{\mathcal{A}}_n=\mathcal{A}_{n,i}$ and the corresponding power allocation can be obtained from~\eqref{E:OptPA}. An outline of the above algorithm is given in~\cref{A:GreedyRRHSel}.
\begin{table}[h]
\caption{Greedy algorithm for problem~\eqref{P:DualFuncnFixUA}}\label{A:GreedyRRHSel}
\begin{framed}
\begin{algorithmic}[1]
\State Initialization: Iteration $i=1$, set of selected RRHs $\mathcal{A}_{n,1}=\emptyset$, maximum objective value $\tilde{f}_1=0$
\For {each $i=1,\dotsc,M$}
\State Find the RRH $j_{i}\in\mathcal{M}\setminus\mathcal{A}_{n,i}$, according to~\eqref{E:MaxRRHji}
\If {$j_i$ satisfies condition~\eqref{E:NewObjGreater}}
\State Update selected set as $\mathcal{A}_{n,i+1}=\mathcal{A}_{n,i}\cup\left\{j_{i}\right\}$
\State Update maximum objective value $\tilde{f}_{i+1}$ according to~\cref{E:UpdCObj} 
\Else 
\State Stop and return RRH set $\check{\mathcal{A}}_n=\mathcal{A}_{n,i}$ and corresponding power allocation given by~\eqref{E:OptPA}
\EndIf
\EndFor
\end{algorithmic}
\end{framed}
\end{table}
From~\eqref{E:MaxRRHji}, it can be seen that each iteration $i$ of the greedy algorithm involves a search over the set $\mathcal{M}\setminus\mathcal{A}_{n,i}$, which has size $M-(i-1)$, where $i\leq M$. As there can be at most $M$ iterations, the greedy algorithm requires $\sum_{i=1}^{M}M-i+1=M(M+1)/2$ iterations in the worst case, which 
is of $O\left(M^2\right)$. Thus, using the greedy algorithm to solve problem~\eqref{P:DualFuncnFixUA} on each SC reduces the overall complexity of solving problem~\eqref{P:Main} to $O\left(NKM^2\right)$.

Notice that the greedy algorithm would recover the optimal solution to problem~\eqref{P:DualFuncnFixUA} in cases where the optimal RRH selection consists of at most two RRHs. However, in general, for given dual variables $\lambda$ and $\bm\mu$, it gives only a suboptimal solution to problem~\eqref{P:DualFuncnFixUA}. This implies that convergence to the optimal dual variables $\lambda^\star$ and $\bm\mu^\star$ cannot be guaranteed if the greedy algorithm of~\cref{A:GreedyRRHSel} is used in~\cref{AL:ORRHSelP} of the overall algorithm for problem~\eqref{P:Main} given in~\cref{A:Overall}. In this case, if the power allocation obtained is not feasible, it can be made feasible by scaling each of the power constraints in~\eqref{C:AvgPCMain}. Similarly, if the constraint in~\eqref{E:TSComb} is not satisfied, the RRHs can be de-selected in increasing order of their contribution to the rate on each SC, until~\eqref{E:TSComb} is satisfied. However, in~\cref{Sec:SimResults}, it is shown through extensive simulations that there is only a negligible difference between the performance of the greedy algorithm and the exhaustive search for $\bm\alpha_n$'s in practical scenarios. This can be understood as follows. The first term in the objective of sub-problem~\eqref{P:DualFuncnFixUA} is the product of a weighting factor and the achievable rate on the SC. As more RRHs are chosen to transmit on a given SC, the weighting factor decreases in linear steps, while the increase in the rate on each SC is only logarithmic. Intuitively, this implies that selecting the first few RRHs correctly is crucial in maximizing the objective of problem~\eqref{P:DualFuncnFixUA}, for which even a greedy search may be sufficient, thus implying the good performance of the greedy algorithm. Thus, in practice, the joint resource allocation problem~\eqref{P:Main} may be solved close to optimally at an even lower complexity of $O(NKM^2)$ compared to $O\left(NK2^M\right)$ for exhaustive search over the RRHs on each SC. 
\section{Simulation Results}\label{Sec:SimResults}
For the simulation setup, we consider a cluster of $M$ RRHs and $K$ users, both of which are distributed uniformly within a circular area of radius $500$~m, and whose center is situated at a distance of $2$~km from the CP. The mmWave wireless fronthaul of bandwidth $W$~Hz and centered at a frequency of $73$~GHz, is shared among the RRHs via TDMA. The channel from the CP to each RRH is LoS, with the path loss given by $69.7 + 24\log_{10}\left(D_m\right)$~dB~\cite{rappaport-etal2015wideband,maccartney-etal2016millimeter}. The CP is assumed to transmit at a fixed power of $46$~dBm~\cite{3gpp36931} with an antenna gain of $27$~dB~\cite{rappaport-etal2015wideband}. 

The wireless access channel is centered at a frequency of $2$~GHz and has a bandwidth $B=20$~MHz, following the 3GPP LTE-A standard~\cite{3gpp36211}, and is divided into $N=128$ SCs using OFDMA. The combined path loss and shadowing~(large-scale fading) is modeled as $38+30\log_{10}\left(d_{k,m}\right)+X$ in dB~\cite{3gpp36931}, where $d_{k,m}$ in meters is the distance between the RRH $m$ and user $k$, and $X$ in dB is the shadowing random variable, which follows a zero-mean Gaussian distribution with a standard deviation of $6$~dB. The multipath channel for the wireless access is modeled using an exponential power delay profile with $\lceil N/4\rceil$ taps, and the small-scale fading on each tap is assumed to follow the Rayleigh distribution. The maximum transmit power at each RRH is set as $\bar{P}_m=24$~dBm, $m\in\mathcal{M}$, and an antenna gain of $2$~dB is assumed, following the 3GPP LTE pico cell parameter specifications~\cite{3gpp36931}. The noise power spectral density is $-174$~dBm/Hz with a noise figure of $7$~dB at all the receivers. For simplicity, we consider maximization of the sum rate in problem~\eqref{P:Main}, i.e., the user rate weights $\omega_k=1,\enspace\forall k\in\mathcal{K}$, and the values are averaged over $5$ random network layouts and $20$ channel realizations for each layout. We compare the performance of the following benchmark schemes with the proposed solutions in~\cref{Sec:PropSol}. 
\begin{itemize}
\item \textbf{Benchmark scheme 1: Single RRH selection}. 
In this scheme, only one RRH is selected on each SC, and there is no coherent-combining gain due to joint transmission by more than one RRHs. 
This is achieved by solving problem~\eqref{P:Main} with the following additional constraints on the RRH selections, 
\begin{align}
\bm 1^\mathsf{T}\bm\alpha_n\leq 1\quad\forall n\in\mathcal{N}.
\end{align} 
The sub-problem~\eqref{P:DualFuncnFixUA} on each SC is then solved by searching for the best among $M$ RRHs that can maximize the objective in~\eqref{P:DualFuncnFixUA}, under the optimal power allocation given by~\eqref{E:OptPASingleRRH}. The optimal dual variables $\lambda^\star$ and $\bm\mu^\star$ for this new problem are found using the algorithm in~\cref{A:Overall}, as before. The overall complexity in this case is thus reduced to $O\left(NKM\right)$.
\item \textbf{Benchmark scheme 2: Equal power allocation}.
In this scheme the transmit power allocations at all the RRHs are fixed as $p_{m,n}=\bar{P}_m/N\enspace\forall m\in\mathcal{M},n\in\mathcal{N}$. With this power allocation, problem~\eqref{P:Main} is solved for the optimal RRH selections $\left\{\bm\alpha_n\right\}$ and the user associations $\left\{\bm\nu_n\right\}$. The sub-problem~\eqref{P:DualFuncnFixUA} is solved by exhaustively searching over the $2^M$ RRH selections. As the transmit power allocations are fixed, the optimal dual variable $\lambda^\star$ corresponding to the constraint in~\eqref{E:TSComb} can be found by a simple bisection search over the interval $0\leq\lambda\leq\max_{m\in\mathcal{M}}R_m$. However, the worst-case complexity of this scheme is still $O\left(NK2^M\right)$. 
\item \textbf{Benchmark scheme 3: Conventional OFDMA}. 
In this scheme, we consider a conventional OFDMA-based system, where each RRH $m$ is pre-assigned a fixed set of SCs denoted by $\mathcal{S}_m\subseteq\mathcal{N}$, where $\left|\mathcal{S}_m\right|=\lfloor N/M\rfloor$. 
It is assumed that each user is associated to its nearest RRH. Thus, each RRH $m\in\mathcal{M}$ transmits to the set of users $\mathcal{K}_m\subseteq\mathcal{K}$ associated to it over the set of SCs $\mathcal{S}_m$, and only one RRH transmits on each SC. As conventional OFDMA systems are typically assumed to have infinite fronthaul capacity, in order to illustrate the effect of a shared fronthaul, we assume that the CP transmits to each RRH for an equal amount of time over the TDMA-based mmWave fronthaul, i.e., $t_m=1/M,\enspace\forall m\in\mathcal{M}$. Then, the fronthaul time-sharing constraint in~\eqref{E:TSComb} is decoupled into individual constraints at each RRH, given by
\begin{align}
\frac{1}{R_m}\sum_{n\in\mathcal{S}_m}\sum_{k\in\mathcal{K}_m}\nu_{k,n}r_{k,n}\left(p_{m,n}\right)\leq\frac{1}{M},\quad\forall m\in\mathcal{M},
\end{align} 
where 
\begin{align}
r_{k,n}\left(p_{m,n}\right)=\frac{B}{N}\log_2\left(1+\frac{\left|h_{k,m,n}\right|^2p_{m,n}}{\sigma^2}\right).
\end{align} 
Note that $p_{m,n}=0\enspace\forall n\notin\mathcal{S}_m,\forall m\in\mathcal{M}$. The joint resource allocation problem~\eqref{P:Main} is then decoupled into $M$ parallel joint power allocation and user association problems, one at each RRH. The optimal solution to this problem for each RRH can be found using the algorithm in~\cref{A:Overall}. Since the RRH-SC selections are already fixed, the optimization is only over the power allocation and the user-SC association at each RRH. The optimal power allocation is given by~\eqref{E:OptPASingleRRH}, and the user on each SC can be found by searching over the $\left|\mathcal{K}_m\right|$ possible users associated with each RRH $m$. Thus, in this case, $M$ problems must be solved in parallel, each with a worst-case complexity given by $O\left(\lfloor \frac{N}{M}\rfloor\left|\mathcal{K}_m\right|\right)$. Thus, if $K_{\mathrm{max}}=\max_{m\in\mathcal{M}}\left|\mathcal{K}_m\right|$ denotes the maximum number of users assigned to an RRH, the effective worst-case complexity of this scheme is given by $O\left(NK_{\mathrm{max}}\right)$. 
\end{itemize}
Both benchmark schemes 1 and 2 require the knowledge of all the channel gains at the CP, where the optimization is performed, and thus, their overhead is same as that for the proposed schemes. On the other hand, for benchmark scheme 3, each RRH needs to know only the channels to its own associated users, and the optimization can be performed at each RRH. 
\begin{figure}[h]
\centering
\includegraphics[width=\linewidth]{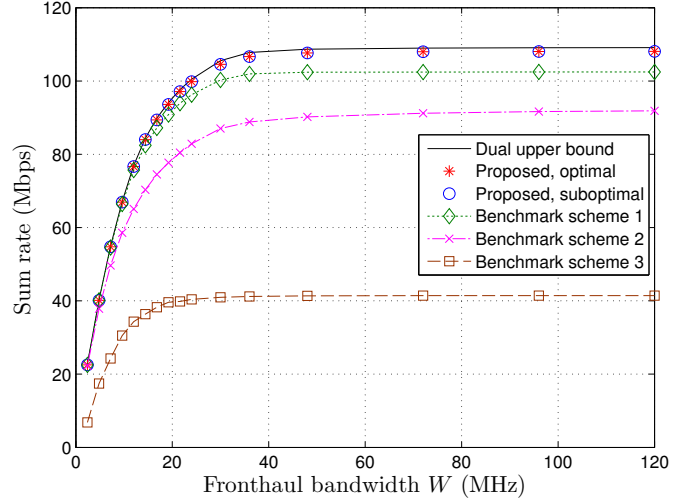}
\caption{Sum rate~(Mbps) vs.\ fronthaul bandwidth $W$~(MHz) for system with $M=6$, $K=8$, $B=20$~MHz and $N=128$.}\label{F:SRvFHW}
\end{figure}

\cref{F:SRvFHW} plots the sum rate achievable by the various schemes against the mmWave fronthaul bandwidth $W$. The dual upper bound given by the optimal value of the dual function $g\left(\lambda^\star,\bm\mu^\star\right)$ is also shown for comparison. 
From~\cref{F:SRvFHW}, it can be observed that the sum rate achieved by the proposed optimal solution is nearly equal to the dual upper bound, thus validating its optimality for large $N$. Moreover, the proposed suboptimal solution of lower complexity introduced in~\cref{SS:SubOptSubP} is observed to perform almost as well as the optimal solution that requires an exhaustive search to find the optimal RRH selection. This can be attributed to the submodularity of the SNR on each SC under the optimal power allocation, which implies that selecting the best few RRHs is most important in achieving the major part of the coherent-combining gain on each SC, and finding the optimal RRH selection may not always be necessary, as explained in~\cref{SS:SubOptSubP}. Thus, in practical deployments with a moderate number of RRHs, a close to optimal solution can be obtained 
at a much lower complexity of $O\left(NKM^2\right)$. 

When the fronthaul bandwidth $W$ is small, benchmark scheme~1 that selects only the best RRH on each SC combined with optimal power allocation, performs almost as well as the optimal solution. In this case, the system is primarily limited by the fronthaul, and only one or very few RRHs can be supported by the fronthaul on each SC. However, the proposed solution outperforms benchmark scheme~2 that optimally selects RRHs under equal power allocation on all SCs at all values of $W$, showing the importance of optimal power allocation. At higher values of $W$, almost all the RRHs are selected on all the SCs, providing a coherent-combining gain on each SC, which is not achieved by the single RRH selection. However, the single RRH selection can still provide a diversity gain, since the best RRH is selected on each SC. Also notice that the performance of benchmark scheme~1 remains constant for $W$ larger than $30$~MHz, which is expected since benchmark scheme~1 selects only one RRH on each SC, and increasing $W$ does not help in achieving a coherent-combining gain on any SC. On the other hand, even with equal power allocation, the performance of benchmark scheme~2 increases slightly with increasing $W$, since in this case, more and more RRHs can be selected on each SC. 

\begin{figure}[h]
\centering
\includegraphics[width=\linewidth]{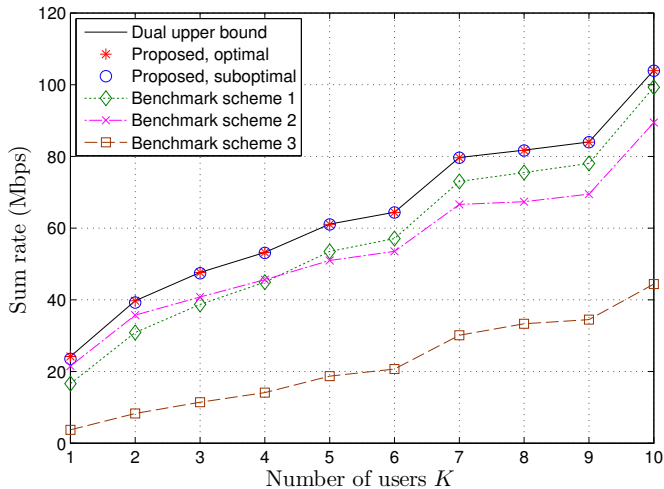}
\caption{Sum rate~(Mbps) vs.\ number of users $K$ for system with $W=50$~MHz, $M=5$, $B=20$~MHz and $N=128$.}\label{F:SRvK}
\end{figure}
\cref{F:SRvK} shows the variation of the sum rate with the number of users, $K$. Here, the fronthaul bandwidth $W=50$~MHz, and the number of RRHs $M=5$, are fixed. The performance comparison of the various schemes is observed to be consistent with that in~\cref{F:SRvFHW}, except for the case of low values of $K$, where benchmark scheme~2 with equal power allocation performs slightly better than benchmark scheme~1 that selects only the best RRH on each SC. This implies that when the number of users is small compared to the number of RRHs, i.e., the network is dense in the RRHs, achieving a coherent-combining gain even with equal power allocation can be better than selecting only the best RRH on each SC.

\begin{figure}[h]
\centering
\includegraphics[width=\linewidth]{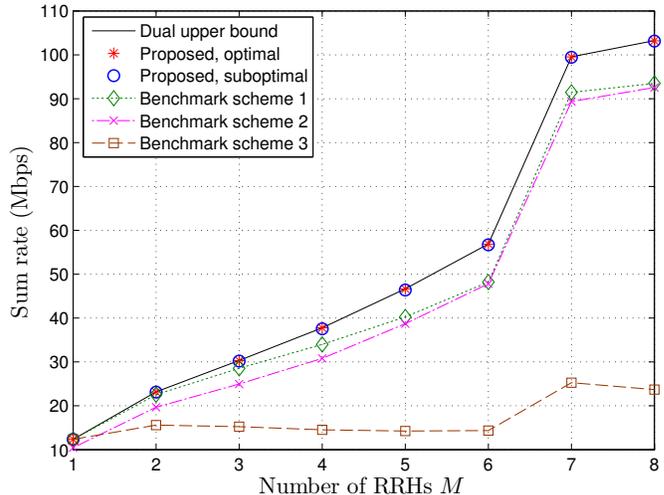}
\caption{Sum rate~(Mbps) vs.\ number of RRHs $M$ for system with $W=100$~MHz, $K=4$, $B=20$~MHz and $N=128$.}\label{F:SRvM}
\end{figure}
\cref{F:SRvM} shows the variation of the sum rate with the number of RRHs, $M$. Here, the fronthaul bandwidth $W=100$~MHz and the number of users $K=4$ are fixed. Again, the trends are similar to~\cref{F:SRvFHW,F:SRvK}. Observe that when $M=1$, both the proposed schemes and benchmark schemes~1 and 3 are equivalent. As $M$ is increased, the performance of the conventional OFDMA with benchmark scheme~3 remains more or less the same, since both an equal time allocation on the fronthaul and an equal division of SCs on the access are performed, which cannot exploit the diversity offered by larger $M$. Also, similar to~\cref{F:SRvK}, as $M$ becomes much larger than $K$, the performance of benchmark scheme~2 with equal power allocation approaches that of benchmark scheme~1 with single RRH selection.

In summary, \cref{F:SRvFHW,F:SRvK,F:SRvM} show the sum-rate gains achieved by the proposed solutions compared to other benchmark schemes. In particular, the sum rate achieved by the proposed solutions is more than double of that of the conventional OFDMA with equal time allocation on the fronthaul, showing the advantage of joint dynamic resource allocation over the mmWave fronthaul and wireless access achieved by the proposed centralized scheduling. 
\section{Conclusion}\label{Sec:Conc}
In this paper, we have studied the downlink transmission in a new OFDMA-based UD-CRAN enabled by the mmWave fronthaul. Specifically, we considered a system where the user assigned on any frequency SC can potentially be served by multiple RRHs, subject to the fronthaul rate constraint. We formulated a joint fronthaul time allocation, RRH-SC selection and power allocation, and user-SC association problem to maximize the WSR of users. Although the problem is combinatorial and non-convex in general, we proposed efficient solutions based on the Lagrange duality technique and greedy search. 

Through numerical simulations, we have shown that both the proposed solutions achieve the optimal throughput performance, and significantly outperform the other benchmark schemes considered, under a practical UD-CRAN setup with mmWave based wireless fronthaul. In particular, our proposed solutions for the OFDMA-based UD-CRAN can achieve throughput gains of more than 150\% over a conventional LTE-A network where each user is associated with a single RRH/BS and the mmWave fronthaul bandwidth is equally divided among the RRHs. Thus, the proposed OFDMA-based UD-CRAN with mmWave fronthaul is a cost-effective and scalable architecture for future 5G networks, and when combined with our proposed resource allocation algorithms, it can provide significant throughput gains, especially over the current LTE-A networks. 
\appendices
\section{Proof of~\Cref{L:rknConc}}\label{App:ProofrknConc}
Without loss of generality, assume $\bm\alpha_n=\bm 1$, and $p_{m,n}>0,~\forall m\in\mathcal{M}$, since otherwise we could equivalently define $r_{k,n}\left(\bm\alpha_n,\bm p_n\right)$ by excluding the RRHs for which $\alpha_{m,n}=0$ or $p_{m,n}=0$ from the summation in~\eqref{E:SCRate}. 
Consider the function 
\begin{align}
f(\bm p)=\left(\sum_{m=1}^M\sqrt{p_m}\right)^2,\quad\bm p\succ\bm 0.\label{E:Func}
\end{align}
The second order partial derivatives of $f(\bm p)$ are given by
\begin{align}
\md{f(\bm p)}{2}{p_i}{}{p_j}{}&=\frac{1}{2\sqrt{p_ip_j}}\quad i\neq j,~i,j\in\mathcal{M}\label{E:Partiali}\\
\pd[2]{f(\bm p)}{p_i}&=\frac{1}{2p_i}-\frac{1}{2p_i^{3/2}}\left(\sum_{m=1}^M\sqrt{p_m}\right)\quad i\in\mathcal{M}.\label{E:Partialij}
\end{align}
Using~\eqref{E:Partiali} and~\eqref{E:Partialij}, the Hessian matrix of $f(\bm p)$ can be written as 
\begin{align}
&\nabla^2f(\bm p)
=\frac{1}{2}\left[\bm u\bm u^\mathsf{T}
-\left(\sum_{m=1}^M\sqrt{p_m}\right)\mathsf{diag}\begin{pmatrix}
p_1^{-\frac{3}{2}}&\cdots&p_M^{-\frac{3}{2}}
\end{pmatrix}\right],
\end{align}
where $\bm u\triangleq\begin{bmatrix}1/\sqrt{p_1}&\cdots&1/\sqrt{p_M}\end{bmatrix}^\mathsf{T}$. Then, for any 
vector 
$\bm v\in\mathds{R}^{M\times 1}$, we have
\begin{align}
&\bm v^\mathsf{T}\nabla^2f(\bm p)\bm v\notag\\
&=\frac{1}{2}\left[\left(\bm u^\mathsf{T}\bm v\right)^2-\left(\sum_{m=1}^M\sqrt{p_m}\right)\bm v^\mathsf{T}\mathsf{diag}\begin{pmatrix}
p_1^{-\frac{3}{2}}&\cdots&p_M^{-\frac{3}{2}}\end{pmatrix}\bm v\right].\label{E:QuadFormHess}
\end{align}
Now define the vectors $\bm a\triangleq\begin{bmatrix}p_1^{1/4}&\cdots&p_M^{1/4}\end{bmatrix}^\mathsf{T}\in\mathds{R}_+^{M\times 1}$, and $\bm b\triangleq\begin{bmatrix}v_1/p_1^{3/4}&\cdots&v_M/p_M^{3/4}\end{bmatrix}^\mathsf{T}\in\mathds{R}^{M\times 1}$.  
Then, using~\cref{E:QuadFormHess} and the Cauchy-Schwarz inequality $\left(\bm a^\mathsf{T}\bm b\right)^2\leq\left\|\bm a\right\|^2\left\|\bm b\right\|^2$, it can be verified that $\bm v^\mathsf{T}\nabla^2f(\bm p)\bm v\leq 0,\enspace\forall \bm v\in\mathds{R}^{M\times 1}$, which implies that the matrix $\nabla^2f(\bm p)$ is negative semidefinite~\cite{boyd2004convex}. Hence $f(\bm p)$ is jointly concave in $\bm p$. Setting $\bm\alpha_n=\bm 1$ in~\eqref{E:SCSNR}, the SNR on SC $n\in\mathcal{N}_k$ can be written as  
\begin{align}
\gamma_{k,n}\left(\bm p_n\right)
&=\left(\sum_{m=1}^M\left(\frac{\left|h_{k,m,n}\right|^2}{\sigma^2}p_{m,n}\right)^{1/2}\right)^2\nonumber\\
&=f\left(\mathsf{diag}\left(\begin{array}{ccc}\frac{\left|h_{k,1,n}\right|^2}{\sigma^2}&\cdots&\frac{\left|h_{k,M,n}\right|^2}{\sigma^2}\end{array}\right)\bm p_n\right)\label{E:SCSNRTrans}
\end{align}
where $f(\bm p_n)$ is defined in~\eqref{E:Func}. Since $f(\bm p_n)$ is jointly concave in $\bm p_n$, and from~\eqref{E:SCSNRTrans}, $\gamma_{k,n}\left(\bm p_n\right)$ is the composition of $f\left(\bm p_n\right)$ with a linear transformation of $\bm p_n$, it follows that $\gamma_{k,n}\left(\bm p_n\right)$ is also jointly concave in $\bm p_n$. Now, the logarithm function is concave and its extended value extension on the real line is non-decreasing. Thus, $r_{k,n}\left(\bm p_n\right)=\frac{B}{N}\log_2\left(1+\gamma_{k,n}\left(\bm p_n\right)\right)$ is the composition of the concave function $\gamma_{k,n}\left(\bm p_n\right)$ with a concave and non-decreasing function, and hence, is also jointly concave in $\bm p_n$~\cite{boyd2004convex}. The proof of~\cref{L:rknConc} is thus completed. 
\section{Proof of~\Cref{Prop:OptPAFixURRHSel}}\label{A:ProofPAlloc}
Consider the following two cases based on whether the factor $F_{\hat{k}_n,n}\left(\tilde{\bm\alpha}_n\right)$ defined in~\eqref{E:FHParam} is either greater than zero, or less than or equal to zero, respectively.
\subsection{Case 1: $F_{\hat{k}_n,n}\left(\tilde{\bm\alpha}_n\right)\leq 0$}
In this case, since $r_{\hat{k}_n,n}\left(\tilde{\bm\alpha}_n,\bm p_n\right)$ is jointly concave in $\bm p_n$ according to~\Cref{L:rknConc}, it follows that the objective of problem~\eqref{P:DualFuncnFixUARRHSel} is jointly convex in $\bm p_n$. Thus, problem~\eqref{P:DualFuncnFixUARRHSel} is non-convex under this condition. However, since $r_{\hat{k}_n,n}\left(\tilde{\bm\alpha}_n,\bm p_n\right)\geq 0$, $\bm\mu\succeq\bm 0$, and $\bm p_n\succeq\bm 0$, it implies that for any non-zero power allocation $\bm p_n\neq\bm 0$, the objective of problem~\eqref{P:DualFuncnFixUARRHSel} is strictly negative, while the objective is zero for $\bm p_n=\bm 0$. Thus, the optimal power allocation for problem~\eqref{P:DualFuncnFixUARRHSel} in this case is to allot zero power on all RRHs, i.e. $\tilde{\bm p}_n=\bm 0$. 
\subsection{Case 2: $F_{\hat{k}_n,n}\left(\tilde{\bm\alpha}_n\right)>0$} 
In this case, the objective of problem~\eqref{P:DualFuncnFixUARRHSel} is jointly concave in $\bm p_n$, as $r_{\hat{k}_n,n}\left(\tilde{\bm\alpha}_n,\bm p_n\right)$ is jointly concave in $\bm p_n$ according to~\Cref{L:rknConc}. Thus problem~\eqref{P:DualFuncnFixUARRHSel} is convex, which implies that there exists a unique $\tilde{\bm p}_n\succeq\bm 0$ that attains the maximum of the objective in~\eqref{P:DualFuncnFixUARRHSel}. Taking the derivative of the objective in~\eqref{P:DualFuncnFixUARRHSel} with respect to the power allocation variable $p_{i,n}$, for each RRH $i\in\mathcal{M}$, and setting it equal to zero gives 
\begin{align}
&\frac{B}{N\ln 2}\left(\omega_{\hat{k}_n}-\lambda\sum_{m=1}^M\frac{\tilde{\alpha}_{m,n}}{R_m}\right)\notag\\
&
\cdot\left[\frac{\frac{1}{\sigma}\left(\sum_{m=1}^M\left|h_{\hat{k}_n,m,n}\right|\tilde{\alpha}_{m,n}\sqrt{p_{m,n}}\right)}{1+\frac{1}{\sigma^2}\left(\sum_{m=1}^M\left|h_{\hat{k}_n,m,n}\right|\tilde{\alpha}_{m,n}\sqrt{p_{m,n}}\right)^2}\right]\notag\\
&
\cdot\frac{\left|h_{\hat{k}_n,i,n}\right|\tilde{\alpha}_{i,n}}{\sigma}p_{i,n}^{-1/2}
-\mu_i
=0.\label{E:DiffSULagpi}
\end{align}
Note that~\eqref{E:DiffSULagpi} must be satisfied by each of the optimal power allocations $\tilde{p}_{i,n}$ on RRH $i\in\mathcal{M}$. The received SNR at the user $\hat{k}_n$ corresponding to this optimal power allocation $\tilde{\bm p}_n$ is given by 
\begin{align}
\tilde{\gamma}_{\hat{k}_n,n}&\triangleq\gamma_{\hat{k}_n,n}\left(\tilde{\bm\alpha}_n,\tilde{\bm p}_n\right)\notag\\
&=\frac{1}{\sigma^2}\left(\sum_{m=1}^M\left|h_{\hat{k}_n,m,n}\right|\tilde{\alpha}_{m,n}\sqrt{\tilde{p}_{m,n}}\right)^2.\label{E:OptSNRSU}
\end{align}
Then, using~\eqref{E:OptSNRSU} in~\eqref{E:DiffSULagpi}, re-arranging and squaring, the optimal power allocation can be written as 
\begin{align}
\tilde{p}_{i,n}&=\left[\frac{B}{N\ln 2}\left(\omega_{\hat{k}_n}-\lambda\sum_{m=1}^M\frac{\tilde{\alpha}_{m,n}}{R_m}\right)\right]^2
\frac{\tilde{\gamma}_{\hat{k}_n,n}}{\left(1+\tilde{\gamma}_{\hat{k}_n,n}\right)^2}\notag\\
&\quad\cdot\frac{\tilde{\alpha}_{i,n}\left|h_{\hat{k}_n,i,n}\right|^2}{\sigma^2\mu_i^2},\quad i\in\mathcal{M}.\label{E:OptPA1SU}
\end{align}
Substituting the value of $\tilde{p}_{i,n}$ from~\eqref{E:OptPA1SU} for each $i\in\mathcal{M}$ in~\eqref{E:OptSNRSU}, the optimal SNR at the user must satisfy the relation 
\begin{align}
\tilde{\gamma}_{\hat{k}_n,n}
&=\left[\frac{B}{N\ln 2}\left(\omega_{\hat{k}_n}-\lambda\sum_{m=1}^M\frac{\tilde{\alpha}_{m,n}}{R_m}\right)\right]^2
\frac{\tilde{\gamma}_{\hat{k}_n,n}}{\left(1+\tilde{\gamma}_{\hat{k}_n,n}\right)^2}\notag\\
&\quad\cdot\left(\sum_{m=1}^M\frac{\tilde{\alpha}_{m,n}\left|h_{\hat{k}_n,m,n}\right|^2}{\sigma^2\mu_m}\right)^2.\label{E:OptSNRSUCond}
\end{align}
Since all the quantities in~\eqref{E:OptSNRSUCond} are non-negative, taking square roots on both sides and re-arranging using the definitions~\eqref{E:FHParam} and~\eqref{E:AccParam} yields
\begin{align}
\tilde{\gamma}_{\hat{k}_n,n}^{1/2}
\left[1+\tilde{\gamma}_{\hat{k}_n,n}-\frac{B}{N\ln 2}F_{\hat{k}_n,n}\left(\tilde{\bm\alpha}_n\right)
G_{\hat{k}_n,n}\left(\tilde{\bm\alpha}_n\right)\right]=0.
\end{align}
The above implies that either $\tilde{\gamma}_{\hat{k}_n,n}=0$ or
\begin{align}
\tilde{\gamma}_{\hat{k}_n,n}=\frac{B}{N\ln 2}F_{\hat{k}_n,n}\left(\tilde{\bm\alpha}_n\right)G_{\hat{k}_n,n}\left(\tilde{\bm\alpha}_n\right)-1.\label{E:OptSNRSUFinal}
\end{align}
If $\tilde{\gamma}_{\hat{k}_n,n}=0$, then from~\eqref{E:OptPA1SU}, it follows that $\tilde{\bm p}_n=\bm 0$. Thus, for a non-zero power allocation, we require $\tilde{\gamma}_{\hat{k}_n,n}>0$ , which from~\eqref{E:OptSNRSUFinal}, translates to the condition,
\begin{align}
\frac{B}{N\ln 2}F_{\hat{k}_n,n}\left(\tilde{\bm\alpha}_n\right)G_{\hat{k}_n,n}\left(\tilde{\bm\alpha}_n\right)-1>0.\label{E:CondSNRSCSU}
\end{align}
In Case 1, $F_{\hat{k}_n,n}\left(\tilde{\bm\alpha}_n\right)\leq 0$, and hence condition~\eqref{E:CondSNRSCSU} cannot be satisfied, since $G_{\hat{k}_n,n}\left(\tilde{\bm\alpha}_n\right)\geq 0$. Since $\tilde{\bm p}_n=\bm 0$ in Case 1, condition~\eqref{E:CondSNRSCSU} subsumes Case 1. Finally, using~\eqref{E:OptSNRSUFinal} in~\eqref{E:OptPA1SU}, along with the definitions~\cref{E:FHParam,E:AccParam} and the condition~\eqref{E:CondSNRSCSU}, the optimal power allocation that solves problem~\eqref{P:DualFuncnFixUARRHSel} is obtained as given in~\eqref{E:OptPA}. The proof of~\cref{Prop:OptPAFixURRHSel} is thus completed. 
\section{Proof of~\cref{Prop:OptSNRSubMod}}\label{App:ProofOptSNRSubMod}
For convenience, we drop the user and SC subscripts $\hat{k}_n$ and $n$ in this proof. Notice that in order to show submodularity of $\tilde{\gamma}(\mathcal{A})$, it suffices to show that the product of the set functions $F(\mathcal{A})G(\mathcal{A})$ is submodular. Now, if RRH $i$ is added to the set $\mathcal{A}$, this product becomes 
\begin{align}
&F\left(\mathcal{A}\cup\{i\}\right)G\left(\mathcal{A}\cup\{i\}\right)\notag\\
&=\left(\omega-\sum_{m\in\mathcal{A}}\frac{\lambda}{R_m}-\frac{\lambda}{R_i}\right)\left(\sum_{m\in\mathcal{A}}\frac{\left|h_m\right|^2}{\sigma^2\mu_m}+\frac{\left|h_i\right|^2}{\sigma^2\mu_i}\right),
\end{align}
which after rearrangement gives
\begin{align}
F\left(\mathcal{A}\cup\{i\}\right)G\left(\mathcal{A}\cup\{i\}\right)&=F\left(\mathcal{A}\right)G\left(\mathcal{A}\right)+\frac{\left|h_i\right|^2}{\sigma^2\mu_i}F\left(\mathcal{A}\right)\notag\\
&\quad-\frac{\lambda}{R_i}G\left(\mathcal{A}\right)-\frac{\lambda\left|h_i\right|^2}{\sigma^2\mu_iR_i}.
\end{align}
Thus, the incremental difference in the value of the product $F\left(\mathcal{A}\right)G\left(\mathcal{A}\right)$ when RRH $i$ is added to $\mathcal{A}$ is given by
\begin{align}
&F\left(\mathcal{A}\cup\{i\}\right)G\left(\mathcal{A}\cup\{i\}\right)-F\left(\mathcal{A}\right)G\left(\mathcal{A}\right)\notag\\
&=\frac{\left|h_i\right|^2}{\sigma^2\mu_i}F\left(\mathcal{A}\right)-\frac{\lambda}{R_i}G\left(\mathcal{A}\right)-\frac{\lambda\left|h_i\right|^2}{\sigma^2\mu_iR_i}.\label{E:IncDiffAAi}
\end{align}
Similarly, the incremental difference when the same RRH $i$ is added to a larger set $\mathcal{A}\cup\{j\}$ can be obtained by replacing the set $\mathcal{A}$ with $\mathcal{A}\cup\{j\}$ and the set $\mathcal{A}\cup\{i\}$ with $\mathcal{A}\cup\{i,j\}$ in~\cref{E:IncDiffAAi}, which gives 
\begin{align}
&F\left(\mathcal{A}\cup\{i,j\}\right)G\left(\mathcal{A}\cup\{i,j\}\right)-F\left(\mathcal{A}\cup\{j\}\right)G\left(\mathcal{A}\cup\{j\}\right)\notag\\
&=\frac{\left|h_i\right|^2}{\sigma^2\mu_i}F\left(\mathcal{A}\cup\{j\}\right)-\frac{\lambda}{R_i}G\left(\mathcal{A}\cup\{j\}\right)-\frac{\lambda\left|h_i\right|^2}{\sigma^2\mu_iR_i}\label{E:IncDiffAijAj}\\
&=F\left(\mathcal{A}\cup\{i\}\right)G\left(\mathcal{A}\cup\{i\}\right)-F\left(\mathcal{A}\right)G\left(\mathcal{A}\right)\notag\\
&\quad-\left(\frac{\lambda\left|h_i\right|^2}{\sigma^2\mu_iR_j}+\frac{\lambda\left|h_j\right|^2}{\sigma^2\mu_jR_i}\right)\label{E:IncDiffRel}\\
&\leq F\left(\mathcal{A}\cup\{i\}\right)G\left(\mathcal{A}\cup\{i\}\right)-F\left(\mathcal{A}\right)G\left(\mathcal{A}\right),\label{E:IncDiffIneq}
\end{align}
where~\eqref{E:IncDiffRel} follows by rearranging~\eqref{E:IncDiffAijAj} and using~\eqref{E:IncDiffAAi}, and the last inequality~\eqref{E:IncDiffIneq} follows since the term in the parentheses in~\eqref{E:IncDiffRel} is non-negative. Thus, $F\left(\mathcal{A}\right)G\left(\mathcal{A}\right)$ satisfies condition~\eqref{E:SubMod}, and is hence submodular, which implies that $\tilde{\gamma}\left(\mathcal{A}\right)$ in~\eqref{E:OptSNRSetFunc} is submodular. This completes the proof of~\cref{Prop:OptSNRSubMod}. 
\bibliographystyle{IEEEtran_mod}
\bibliography{IEEEabrv,bibJournalList,ThesisBibliography}

\vfill

\enlargethispage{5in}

\end{document}